\documentclass[12pt,oneside,reqno]{amsart}
\usepackage{tikz}
\usetikzlibrary{calc}
\usetikzlibrary{decorations.markings,arrows}
\usepackage{subcaption}

\usepackage{amsmath,amsthm,verbatim,amssymb,amsfonts,amscd, graphicx, float,pgfplots}
\usepackage{graphics}
\usepackage{geometry}
\usepackage[colorinlistoftodos]{todonotes}

\usepackage{thmtools}
\usepackage{thm-restate}
\usepackage{hyperref}
\usepackage[noabbrev]{cleveref}
\crefname{subsection}{subsection}{subsections}

\renewcommand{\epsilon}{\varepsilon}

\newtheorem{theorem}{Theorem}[section]
\newtheorem{corollary}[theorem]{Corollary}
\newtheorem{lemma}[theorem]{Lemma}
\newtheorem{proposition}[theorem]{Proposition}

\theoremstyle{definition}
\newtheorem{definition}{Definition}[section]

\theoremstyle{remark}
\newtheorem*{remark}{Remark}

\begin{document}
\title{The Asymmetric Colonel Blotto Game}
\author{Simon Rubinstein-Salzedo}
\address{Simon Rubinstein-Salzedo: Euler Circle, Palo Alto, CA 94306, USA}
\email{simon@eulercircle.com}
\author{Yifan Zhu}
\address{Yifan Zhu: Shanghai Foreign Language School, Shanghai 200083, China}
\email{fanzhuyifan@gmail.com}
\maketitle

\begin{abstract} 
  This paper explores the Nash equilibria of a variant of the Colonel Blotto game, which we call the Asymmetric Colonel Blotto game.
  In the Colonel Blotto game, two players simultaneously distribute forces across $n$ battlefields. 
  Within each battlefield, the player that allocates the higher level of force wins. 
  The payoff of the game is the proportion of wins on the individual battlefields. 
  In the asymmetric version, the levels of force distributed to the battlefields must be nondecreasing.
  In this paper, we find a family of Nash equilibria for the case with three battlefields and equal levels of force and prove the uniqueness of the marginal distributions. 
  We also find the unique equilibrium payoff for all possible levels of force in the case with two battlefields, and obtain partial results for the unique equilibrium payoff for asymmetric levels of force in the case with three battlefields.
\end{abstract}


\section{Introduction}
In this section we discuss the background and origins of the Asymmetric Colonel Blotto game.

The Colonel Blotto game, which originates with Borel in~\cite{Borel53}, is a constant-sum game involving two players, $A$ and $B$, and $n$ independent battlefields. 
$A$ distributes a total of $X_A$ units of force among the battlefields, and $B$ distributes a total of $X_B$ units of force among the battlefields, in such a way that each player allocates a nonnegative amount of force to each battlefield.
The player who sends the higher level of force to a particular battlefield wins that battlefield.
The payoff for the whole game is the proportion of the wins on the individual battlefields.

Roberson in \cite{Roberson2006} characterizes the unique equilibrium payoffs for all (symmetric and asymmetric) configurations of the players' aggregate levels of force, and characterizes the complete set of equilibrium univariate marginal distributions for most of these configurations for the Colonel Blotto game.

A possible variant of the Colonel Blotto game, which has not been studied before, is the Asymmetric Colonel Blotto game, where the forces distributed among the battlefields must be in non-decreasing order.

The Asymmetric Colonel Blotto game is a constant-sum game involving two players, $A$ and $B$, and $n$ independent battlefields. $A$ distributes $X_A$ units of force among the battlefields in a nondecreasing manner and $B$ distributes $X_B$ units of force among the battlefields in a non-decreasing manner. Each player distributes forces without knowing the opponent's distribution. The player who provides the higher amount of force to a battlefield wins that battlefield. If both players deploy the same amount of force to a battlefield, we declare that battlefield to be a draw, and the payoff of that battlefield is equally distributed among the two players.\footnote{As we show in \Cref{thm:noAtom}, Nash equilibria of games with equal levels of force do not contain atoms, so the probability that the two players place equal force on some battlefield is 0. Thus we may, if we choose, use a different tie-breaking rule without altering the result in this case.} 
The payoff for each player is the proportion of battlefields won.

In this paper, we study the Nash equilibria and equilibrium payoffs of Asymmetric Colonel Blotto games.

In \Cref{sec:uniqueCDFACB113}, we find a family of equilibria for the game with three battlefields and equal levels of force, and we prove the uniqueness of the marginal distribution functions.
We also prove that in any equilibrium strategies for a game with equal levels of force and at least three battlefields, there are no atoms in the marginal distributions.

In \Cref{sec:w2t}, we find the unique equilibrium payoffs of all cases of the Asymmetric Colonel Blotto game involving only two battlefields, and in \Cref{sec:w3t} we find the unique equilibrium payoffs in the case of three battlefields in certain cases. We conclude with \Cref{sec:open}, where we discuss the difficulties in extending our work to the case of $n\ge 4$ battlefields.

\section{The model}

In this section we introduce the model and related concepts.
The definitions in this section are adaptations from those in \cite{Roberson2006} to the asymmetric version.

\subsection{Players}
Two players, $A$ and $B$, simultaneously allocate their forces $X_A$ and $X_B$ across $n$ battlefields in a nondecreasing manner.
Each player distributes forces without knowing the opponent's distribution. 
The player who provides the higher level of force to a battlefield wins that battlefield, gaining a payoff of $\frac{1}{n}$.
If both players deploy the same level of force to a battlefield, that battlefield is a draw and both players gain a payoff of $\frac{1}{2n}$.
The payoff for each player is the proportion of battlefields won, or equivalently, the sum of the payoffs across all the battlefields.\footnote{That the payoff for each player is the sum of the payoffs across all the battlefields means that two different joint distributions are equivalent if they have the same marginal distributions. Hence, this definition makes it possible to separate a joint distribution into the marginal distributions and a $n$-copula later in this paper.}

Player $i$ sends $x^k_i$ units to the $k$th battlefield.
For player $i$, the set of feasible allocations of force across the $n$ battlefields in the Asymmetric Colonel Blotto game is denoted by $\mathfrak{B}_i$:
$$\mathfrak{B}_i=\left\{\mathbf{x}\in \mathbb{R}^n\ \middle|\ \sum^n_{j=1}x_i^j = X_i,0\leq x^1\leq x^2\leq\cdots\leq x^n\right\}.$$

\begin{definition}
  Given an $n$-variate cumulative distribution function $H$, for every $\mathbf{x,y}\in \mathbb{R}^n$ such that $x_k\leq y_k$ for all $k\in\{1,\dots,n\}$,
  the \emph{$H$-volume} of the $n$-box $[x_1,y_1]\times\cdots\times[x_n,y_n]$ is,
  $$V_H\left([\mathbf{x,y}]\right)=\underset{n}{\Delta}_{x_n}^{y_n}\underset{n-1}{\Delta}_{x_{n-1}}^{y_{n-1}}\dots\underset{2}{\Delta}_{x_2}^{y_2}\underset{1}{\Delta}_{x_1}^{y_1}H(\mathbf{t}),$$
  where
  \[
  \underset{k}{\Delta}_{x_k}^{y_k}H(\mathbf{t})=H(t_1,\ldots,t_{k-1},y_k,t_{k+1},\ldots,t_n) 
  -H(t_1,\ldots,t_{k-1},x_k,t_{k+1},\ldots,t_n).
  \]

  Intuitively, the $H$-volume of a $n$-box just measures the probability that a point within that $n$-box will be chosen given the cumulative distribution function $H$.
\end{definition}

\begin{definition}
  The \emph{support} of an $n$-variate cumulative distribution function $H$ is the complement of the union of all open sets of $\mathbb{R}^n$ with $H$-volume zero.
  Intuitively, the support of a mixed strategy is just the closure of the set of pure strategies that might be chosen.
\end{definition}

\subsection{Strategies}
A mixed strategy, or a \emph{distribution of force}, for player $i$ is an $n$-variate cumulative distribution function (cdf) $P_i:\mathbb{R}^n_+\rightarrow[0,1]$ with support in the set of feasible allocations of force $\mathfrak{B_i}$.
This means that if player $i$ chooses strategy $(X^j)_{j=1}^n$, then the probability that $X^j\leq x^j$ $(j=1,\dots,n)$ is $P_i(x^1,\dots,x^n)$. 
$P_i$ has \emph{marginal} cumulative distribution functions $\left\{F_i^j\right\}_{j=1}^n$, one univariate marginal cumulative distribution function for each battle field $j$. $F_i^j(x^j)$ is the probability that $X^j\leq x^j$. Equivalently, $F_i^j(x)=P_i(X_i,X_i,\dots,x,X_i,\dots,X_i)$, where the $j$th argument is $x$, and the rest of the arguments are $X_i$, the player's entire allocation of force.
We write $P_i=\left(F_i^j\right)_{j=1}^n$.

In the case where the mixed strategy is the combination of finite pure strategies, the mixed strategy $P_i$ where $(i^1_j,i^2_j,\dots,i^n_j)$ units of force are distributed the battlefields $1,2,\dots,n$ respectively with probability $p_j$ is denoted by
$$P_i=\left\{\left(\left(i^1_j,i^2_j,\dots,i^n_j\right),p_j\right)\right\}.$$
Here $i^1_j+i^2_j+\dots+i^n_j=X_i$ and $\sum_j p_j = 1$.

\subsection{The Asymmetric Colonel Blotto Game}
The \emph{Asymmetric Colonel Blotto game} with $n$ battlefields, denoted by 
$$ACB(X_A,X_B,n),$$
is a one-shot game in which players simultaneously and independently announce distributions of force $(x_i^1,\ldots,x_i^n)$ subject to their budget constraints $\sum_{j=1}^n x_1^j=X_A$ and $\sum_{j=1}^n x_2^j=X_B$, $x_i^j\ge 0$ for each $i,j$, and such that $x_i^1\le x_i^2\le\cdots x_i^n$ for $i=1,2$. Each battlefield, providing a payoff of $\frac{1}{n}$, is won by the player that provides the higher allocation of force on that battlefield (and declared a draw if both players allocate the same level of force to a battlefield, each gaining a payoff of $\frac{1}{2n}$), and players' payoffs equal the sum of the payoffs over all the battlefields.

\subsection{Nash equilibrium}
Mixed strategies $P_A$ and $P_B$ form a Nash equilibrium if and only if neither player can increase payoff by changing to a different strategy. 


Since this particular game is two-player and constant-sum, it has the interesting property that the equilibrium payoff is always unique:

\begin{theorem}
  \label{thm:uniqueNEPayoff}
  The Nash equilibrium payoff for both players of any two-player and constant-sum game is unique. 
\end{theorem}
\begin{proof}
  Suppose $P_A$ and $P_B$ is a pair of Nash equilibrium strategies. 
  Let $w_i$ be the payoff for player $i$.
  For any pair of Nash equilibrium strategies $P'_A$ and $P'_B$, let $w'_i$ be the payoff for player $i$.

  Let us consider the payoff for both players when player $A$ plays strategy $P_A$ and player $B$ plays strategy $P'_{B}$.
  Call the payoff for player $A$ $v_A$ and the payoff for player $B$ $v_B$. 
  Since $P_A$ is a strategy in a Nash equilibrium, $w_A\ge v_A$. Similarly, $w'_B\ge v_B$. So $w_A+w'_B\ge v_A+v_B$. Since we are considering a constant sum game, $w_A+w'_B\ge v_A+v_B=w_A+w_B$. Hence, $w'_B\ge w_B$. Similarly, we must have $w_B\ge w'_B$. So $w_B=w'_B$. Similarly $w_A=w'_A$.
\end{proof}

\section{Optimal univariate marginal distributions for three battlefields}
\label{sec:uniqueCDFACB113}
In this section we use copulas to separate the joint distributions of players into the marginal distributions and suitable copula. We also find and prove the unique univariate marginal distribution for $ACB(1,1,3)$.

Let us first introduce the concept of copulas:
\begin{definition}
  Let $I$ denote the unit interval $[0,1]$. An \emph{$n$-copula} is a function $C$ from $I^n$ to $I$ such that
  \begin{enumerate}
    \item 
      For all $\mathbf{x} \in I^n$, $C(\mathbf{x}) = 0$ if at least one coordinate of $\mathbf{x}$ is $0$; and if all coordinates of $\mathbf{x}$ are $1$ except $x_k$, then $C (\mathbf{x}) = x_k$.
    \item
      For every $\mathbf{x, y}\in I^n$ such that $x_k \le y_k$ for all $k\in\{1,\dots,n\}$, the $C$-volume of the $n$-box $[x_1, y_1] \times \dots\times [x_n, y_n]$ satisfies \[V_C([\mathbf{x},\mathbf{y}])\ge 0.\]
  \end{enumerate}
\end{definition}

The crucial property of $n$-copulas that we need is the following theorem of Sklar:

\begin{theorem}[Sklar,~\cite{Sklar59}]
  \label{thm:Sklar}
  Let $H$ be an $n$-variate distribution function with univariate marginal distribution functions $F_1, F_2, \dots, F_n$. Then there exists an $n$-copula $C$ such that for all $\mathbf{x} \in \mathbb{R}^n$,
  \begin{equation}
    \label{eq:Sklar}
    H (x_1, \dots , x_n) = C (F_1 (x_1) , \dots, F_n (x_n)).
  \end{equation}
  Conversely, if $C$ is an $n$-copula and $F_1,F_2,\dots ,F_n$ are univariate distribution functions, then the function $H$ defined by \cref{eq:Sklar} is an $n$-variate distribution function with univariate marginal distribution functions $F_1, F_2, \dots , F_n$.
\end{theorem}

The proof of this theorem can be found in \cite{SchweizerSklar}.

This theorem establishes the equivalence between a joint distribution on the one hand, and a combination of a complete set of marginal distributions and a $n$-copula on the other hand.
We will now show that the univariate marginal distribution functions and the $n$-copula are separate components of the players' best responses.

\begin{proposition}
  \label{prop:Lagrange}
  In the game $ACB(X_A, X_B , n)$, suppose that the opponent's strategy is fixed as the distribution $P_{-i}$, and that $X_A = X_B$. Then, in order for player $i$ to maximize payoff under the constraint that the support of the chosen strategy must be in $\mathfrak{B}_i$, player $i$ must solve an optimization problem. Given that there are no atoms in Nash equilibrium strategies (\Cref{thm:noAtom}), we can write the Lagrangian for this optimization problem as
  \begin{equation}
    \label{eq:Lagrange}
    \max\limits_{\left\{F_{i}^ j\right\}_{j=1}^n}\lambda_i\sum_{j=1}^n\left[\int_0^\infty \left[ \frac{1}{n\lambda_i}F^j_{-i}(x)-x \right]dF^j_i \right]+\lambda_iX_i,
  \end{equation}
  where the set of univariate marginal distribution functions $\left\{F_{i}^ j\right\}_{j=1}^n$ satisfy the constraint that there exists an $n$-copula $C$ such that the support of the $n$-variate distribution $$C\left( F^1_i\left(x^1\right),\dots,F^n_i\left(x^n\right) \right)$$ is contained in $\mathfrak{B}_i$.\footnote{Here we only maximize over the set of $\left\{F_{i}^ j\right\}_{j=1}^n$ that satisfy the constraint, not all of them.}
\end{proposition}

\begin{proof}
  The payoff for player $i$ given the opponent's marginal distribution functions $\left\{F_{-i}^ j\right\}_{j=1}^n$ is the sum of the payoffs across all the battlefields:

  \[
  \sum_{j=1}^n \int_0^\infty\frac{1}{n}F_{-i}^j(x)\, dF_i^j.
  \] 

  Here, the integral is the Riemann-Stieltjes integral, so the integrand is $0$ for $x>X_i$.
  We also use the Riemann-Stieltjes integral for other integrals later in the paper.

  \[
  \max\limits_{P_i} \sum_{j=1}^n \int_0^\infty\frac{1}{n}F_{-i}^j(x)\,dF_i^j.
  \]
  That $P_i$ is contained in $\mathfrak{B}_i$ implies that the sum of the levels of force across all battlefields is $X_i$:
  \[
  \sum_{j=1}^n \int_0^\infty x\ dF^j_i = X_i.
  \]
  Hence, the Lagrangian is 
  \begin{multline*}
    \max \limits_{P_i} \left[ \sum_{j=1}^n \int_0^\infty\frac{1}{n}F_{-i}^j(x)dF_i^j -\lambda_i\left[ \sum_{j=1}^n \int_0^\infty x\ dF^j_i - X_i \right] \right] \\
    = \max\limits_{\left\{F_{i}^ j\right\}_{j=1}^n}\lambda_i\sum_{j=1}^n\left[\int_0^\infty \left[ \frac{1}{n\lambda_i}F^j_{-i}(x)-x \right]dF^j_i \right]+\lambda_iX_i .
  \end{multline*}
  Finally, from \Cref{thm:Sklar} the $n$-variate distribution function $P_i$ is equivalent to the set of univariate marginal distribution functions $\left\{F_{i}^ j\right\}_{j=1}^n$ combined with an appropriate $n$-copula, $C$, so the result follows directly.
\end{proof}

\begin{restatable}{theorem}{uniqueCDFACB}
  \label{thm:uniqueCDFACB113}
  The unique Nash equilibrium univariate marginal distribution functions of the game $ACB(1, 1 , 3)$ are for each player to allocate forces according to the following univariate distribution functions:  \begin{align*}
    F^1(u)          & = \left\{                        
    \begin{array}{cc}
      3u              & 0\leq u\leq\frac{1}{3}           \\
      1               & \frac{1}{3}<u\leq 1              
    \end{array}
    \right.
    \\
    F^2(u)          & = \left\{                        
    \begin{array}{cc}
      0               & 0\leq u<\frac{1}{6}              \\	
      -\frac{1}{2}+3u & \frac{1}{6}\leq u\leq\frac{1}{2} \\
      1               & \frac{1}{2}<u\leq 1              
    \end{array}
    \right.
    \\
    F^3(u)          & = \left\{                        
    \begin{array}{cc}
      0               & 0\leq u<\frac{1}{3}              \\	
      -1+3u           & \frac{1}{3}\leq u\leq\frac{2}{3} \\
      1               & \frac{2}{3}<u\leq 1              
    \end{array}
    \right.
    \\
  \end{align*}
  The expected payoff for both players is $\frac{1}{2}$.

  This means that any equilibrium strategies must have the marginal distributions described above, and that any joint distribution with support in $\mathfrak{B}_i$ with such marginal distributions is an equilibrium strategy.
\end{restatable}

Intuitively, it is easy to see why this particular set of marginal distributions might guarantee a Nash equilibrium. Since the distribution density is the same among the three battlefields, the payoff of a pure strategy $p = (a,b,c)$ remains constant at $\frac{1}{2}$ when it changes inside the region $0\le a\le\frac{1}{3}$, $\frac{1}{6}\le b\le\frac{1}{2}$, and $\frac{1}{3}\le c\le \frac{2}{3}$. A player can only hope to increase payoff above that given by $p$ by moving below the lower bound of the marginal distribution in some battlefield and staying inside the bounds of the marginal distributions in the other battlefields. However, this is impossible: $a$ cannot be negative; any attempt to bring $b$ below $\frac{1}{6}$ would result in $c$ being above the upper bound $\frac{2}{3}$; $c$, as the biggest of the 3, cannot be below $\frac{1}{3}$.
(The rigorous proof of this can be found in \Cref{lem:payoffPureAgainstP}.)


Before we give the formal proof of this theorem, let us first examine some joint distributions that satisfy the conditions in \Cref{thm:uniqueCDFACB113}.

Consider the $3$-variate distribution function $P_1$ that uniformly places mass $\frac{1}{3}$ on each of the three sides of the equilateral triangle with vertices $\left(\frac{1}{3},\frac{1}{3},\frac{1}{3}\right)$, $\left(0,\frac{1}{2},\frac{1}{2}\right)$, and $\left(\frac{1}{6},\frac{1}{6},\frac{2}{3}\right)$ to $\left(\frac{1}{3},\frac{1}{3},\frac{1}{3}\right)$ (Depicted in \Cref{fig:jointNE01}). Clearly its marginal distributions are those described in \Cref{thm:uniqueCDFACB113}.

Similarly, as in \Cref{fig:jointNE02}, divide the original equilateral triangle into three smaller equilateral triangles with side lengths $\frac{1}{3}$ of the original, and let $P_2$ be the strategy that uniformly distribute on the sides of the smaller triangles. Clearly $P_2$ has the same marginal distributions as $P_1$, and is thus a joint distribution as described in \Cref{thm:uniqueCDFACB113}. 
As shown in \Cref{fig:jointNE03}, we can continue this process on the smaller triangles (or only on some of the smaller triangles), and thus we obtain an countably-infinite family of joint distributions with marginal distributions as described in \Cref{thm:uniqueCDFACB113}.
Furthermore, given any two such suitable joint distributions, their weighted average is also a suitable joint distribution, and thus we obtain a continuum of suitable joint distributions.
\begin{figure}[htb]
  \center
  \begin{subfigure}[t]{0.4\linewidth}
    \center
    \begin{tikzpicture}[scale = 3]
      \coordinate (A) at (0,0);
      \coordinate (B) at (0.5,{sqrt(3)/2});
      \coordinate (C) at (1,0);
      \draw (A) node[anchor=north]{$(0,0,1)$}
      -- (C) node[anchor=north]{$(0,1,0)$}
      -- (B) node[anchor=south]{$(1,0,0)$}
      -- cycle;
      \coordinate (O) at (0.5, {1/2/sqrt(3)});
      \filldraw[fill = gray] (0.5,0) node[anchor = north]{$x_2\le x_3$}
      -- (O) 
      -- (A)
      -- cycle;
      \draw[dash dot] (O) -- (B);
      \draw[dash dot] (O) 
      -- (0.75, {sqrt(3)/4}) node[anchor = south west]{$x_1\le x_2$};
    \end{tikzpicture}
    \caption{The support, $\mathfrak{B}_i$, is the shaded triangle, which is the triangle shown in \Cref{fig:jointNE01,fig:jointNE02,fig:jointNE03}.}
  \end{subfigure}
  ~
  \begin{subfigure}[t]{0.4\linewidth}
    \center
    \begin{tikzpicture}[scale = 8]
      \coordinate (A) at (0,0);
      \coordinate (P1) at (0.5, {1/2/sqrt(3)});
      \coordinate (P2) at (0.5,0);
      \coordinate (P3) at (0.25, {1/4/sqrt(3)});

      \draw (0.5,0) 
      -- (P1) 
      -- (A) node[anchor = north]{$\left( 0, 0, 1 \right)$}
      -- cycle;
      \draw[ultra thick, color = blue] (P1) node[anchor = south]{$\left( \frac{1}{3},\frac{1}{3}, \frac{1}{3} \right)$}
      -- (P2) node[anchor = north]{$\left(0,\frac{1}{2},\frac{1}{2}\right)$}
      -- (P3) node[anchor = south east]{$\left( \frac{1}{6}, \frac{1}{6}, \frac{2}{3} \right)$}
      -- cycle;
    \end{tikzpicture}
    \caption{The strategy that distributes uniformly on the blue lines is an equilibrium strategy.}
    \label{fig:jointNE01}
  \end{subfigure}

  \begin{subfigure}[t]{0.4\linewidth}
    \center
    \begin{tikzpicture}[scale = 10]
      \coordinate (A) at (0,0);
      \coordinate (P1) at (0.5, {1/2/sqrt(3)});
      \coordinate (P2) at (0.5,0);
      \coordinate (P3) at (0.25, {1/4/sqrt(3)});

      \draw (0.5,0) 
      -- (P1) 
      -- (A) node[anchor = north]{$\left( 0, 0, 1 \right)$}
      -- cycle;
      \draw[ultra thick, color = red] (P1) node[anchor = south]{$\left( \frac{1}{3},\frac{1}{3}, \frac{1}{3} \right)$}
      -- ($1/3*(P2)+2/3*(P1)$) node[anchor = west]{\tiny $\left(\frac{2}{9},\frac{7}{18},\frac{7}{18}\right)$}
      -- ($1/3*(P3)+2/3*(P1)$) node[anchor = south east]{\tiny $\left( \frac{5}{18}, \frac{5}{18}, \frac{4}{9} \right)$}
      -- cycle;
      \draw[ultra thick, color = red] (P2) node[anchor = north]{$\left( 0,\frac{1}{2}, \frac{1}{2} \right)$}
      -- ($1/3*(P1)+2/3*(P2)$) node[anchor = west]{\tiny $\left(\frac{1}{9},\frac{4}{9},\frac{4}{9}\right)$}
      -- ($1/3*(P3)+2/3*(P2)$) node[anchor = north east]{\tiny $\scriptstyle \left( \frac{1}{18}, \frac{7}{18}, \frac{5}{9} \right)$}
      -- cycle;
      \draw[ultra thick, color = red] (P3) node[anchor = east]{$\left( \frac{1}{6}, \frac{1}{6}, \frac{2}{3} \right)$}
      -- ($1/3*(P1)+2/3*(P3)$) node[anchor = south]{\tiny $\left(\frac{2}{9},\frac{2}{9},\frac{5}{9}\right)$}
      -- ($1/3*(P2)+2/3*(P3)$) node[anchor = north]{\tiny $\left( \frac{1}{9}, \frac{5}{18}, \frac{11}{18} \right)$}
      -- cycle;
    \end{tikzpicture}
    \caption{The strategy that distributes uniformly on the red lines is an equilibrium strategy.}
    \label{fig:jointNE02}
  \end{subfigure}
  ~
  \begin{subfigure}[t]{0.4\linewidth}
    \center
    \begin{tikzpicture}[scale = 10]
      \coordinate (A) at (0,0);
      \coordinate (P1) at (0.5, {1/2/sqrt(3)});
      \coordinate (P2) at (0.5,0);
      \coordinate (P3) at (0.25, {1/4/sqrt(3)});

      \draw (0.5,0) node[anchor = north]{$\left( 0,\frac{1}{2}, \frac{1}{2} \right)$}
      -- (P1) node[anchor = south]{$\left( \frac{1}{3},\frac{1}{3}, \frac{1}{3} \right)$}
      -- (A) node[anchor = north]{$\left( 0, 0, 1 \right)$}
      -- cycle;
      \foreach \x/\y/\z in {1/2/3,2/3/1,3/1/2,1/3/2,2/1/3,3/2/1}{
      \draw[ultra thick, color = green] {(P\x)} 
      -- ($1/9*(P\y)+8/9*(P\x)$)  
      -- ($1/9*(P\z)+8/9*(P\x)$)
      -- cycle;

      \draw[ultra thick, color = green] ($2/9*(P\y)+7/9*(P\x)$) 
      -- ($3/9*(P\y)+6/9*(P\x)$)  
      -- ($1/9*(P\z)+6/9*(P\x)+2/9*(P\y)$)
      -- cycle;
      }

    \end{tikzpicture}
    \caption{The strategy that distributes uniformly on the green lines is an equilibrium strategy.}
    \label{fig:jointNE03}
  \end{subfigure}
  \caption{Equilibrium strategies.}
  \label{fig:support}
\end{figure}
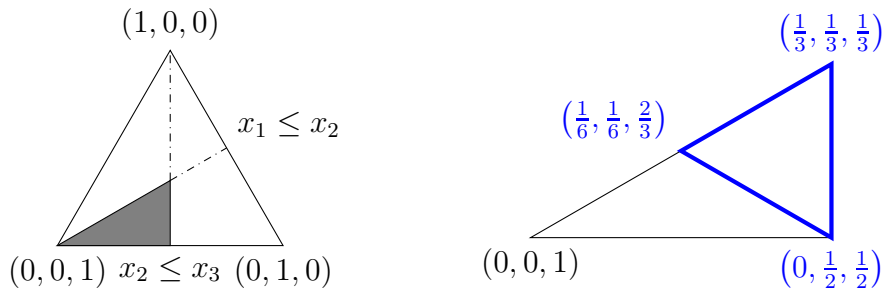
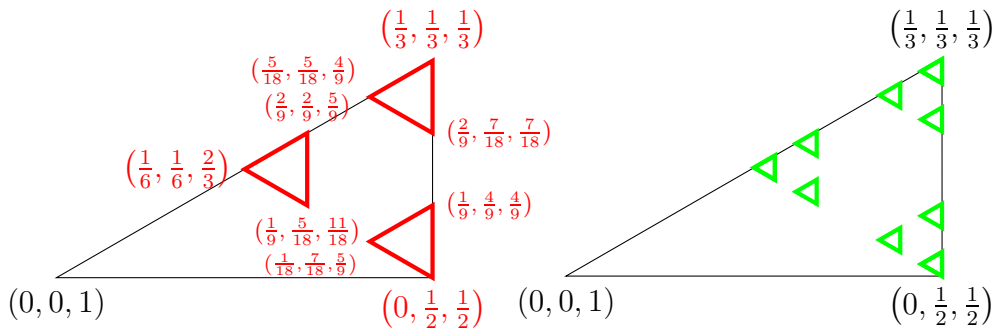

Given these joint distributions that have marginal distributions as characterized by \Cref{thm:uniqueCDFACB113}, we have the following theorem:

\begin{theorem}
  \label{thm:existCopula}
  For the unique set of equilibrium univariate marginal distribution functions $\left\{F_{i}^ j\right\}_{j=1}^3$ characterized in \Cref{thm:uniqueCDFACB113}, there exists a $3$-copula $C$ 
  such that the support of the $3$-variate distribution function 
  $$C\left( F_i^1\left( x^1 \right), F_i^2\left( x^2 \right), F_i^3\left( x^3 \right) \right)$$
  is contained in $\mathfrak{B}_i$.
\end{theorem}

\begin{proof}
  Consider the $3$-variate distribution function $P_1$ that uniformly places mass $\frac{1}{3}$ on each of the three sides of the equilateral triangle with vertices $\left(\frac{1}{3},\frac{1}{3},\frac{1}{3}\right)$, $\left(0,\frac{1}{2},\frac{1}{2}\right)$, and $\left(\frac{1}{6},\frac{1}{6},\frac{2}{3}\right)$ to $\left(\frac{1}{3},\frac{1}{3},\frac{1}{3}\right)$ (depicted in \Cref{fig:jointNE01}). Clearly its marginal distributions are those described in \Cref{thm:uniqueCDFACB113}, and its support is in $\mathfrak{B}_i$.
  Hence, according to Sklar's theorem (\Cref{thm:Sklar}), for the unique set of equilibrium univariate marginal distribution functions $\left\{F_{i}^ j\right\}_{j=1}^3$ characterized in \Cref{thm:uniqueCDFACB113}, there exists a $3$-copula $C$ 
  such that the support of the $3$-variate distribution function $C\left( F_i^1\left( x^1 \right), F_i^2\left( x^2 \right), F_i^3\left( x^3 \right) \right)$ is contained in $\mathfrak{B}_i$.

\end{proof}

Before we provide the formal proof of \Cref{thm:uniqueCDFACB113}, we first seek to provide some intuition for the outline of the proof, which takes inspiration from the proofs in \cite{Roberson2006} and \cite{BayeEtAl}. 

From \cref{eq:Lagrange} in \Cref{prop:Lagrange}, we know that in an Asymmetric Colonel Blotto game $ACB(1,1,3)$, each player's Lagrangian can be written as
\[
\max\limits_{\left\{F_{i}^j\right\}_{j=1}^3}\lambda_i\sum_{j=1}^3\left[\int_0^\infty \left[ \frac{1}{3\lambda_i}F^j_{-i}(x)-x \right]dF^j_i \right]+\lambda_iX_i,
\]
subject to the constraint that there exists an $n$-copula, $C$, such that the support of the $n$-variate distribution $C\left( F^1_i\left(x^1\right),\dots,F^n_i\left(x^n\right) \right)$ is contained in $\mathfrak{B}_i$. 
If there exists a suitable $3$-copula, then, for different $j$, $F_i^j$ is independent. So \cref{eq:Lagrange113} is the maximization of three independent sums, hence the sum of three independent maximizations:
\begin{multline*}
  \max\limits_{\left\{F_{i}^ j\right\}_{j=1}^3}\lambda_i\sum_{j=1}^3\left[\int_0^\infty \left[ \frac{1}{3\lambda_i}F^j_{-i}(x)-x \right]\, dF^j_i \right]+\lambda_iX_i \\
  = \sum_{j=1}^3 \max\limits_{F_{i}^ j}\lambda_i\int_0^\infty \left[ \frac{1}{3\lambda_i}F^j_{-i}(x)-x \right]\,dF^j_i +\lambda_iX_i.
\end{multline*}
Hence we have reduced the maximization problem over a joint distribution to separate maximization problems over univariate distributions, which can be easily solved.

Note that each separate maximization problem has the same form as that of an all-pay auction. 
An all-pay auction is an auction where several players simultaneously call out a bid for a prize, and all bidders pay regardless of who wins the prize; the prize is awarded to the highest bidder.
In an all-pay auction with two bidders, let $F_i$ represent bidder $i$'s distribution of the bid, and $v_i$ represent the value of the auction for bidder $i$. 
Each bidder $i$'s problem is

\[
\max\limits_{ F_i } \int^\infty _0 \left[ v_i F_{-i}(x) -x \right]dF_i.
\]

In the separate maximization problems for the Asymmetric Colonel Blotto game, the quantity $\frac{1}{3\lambda_i}$ acts as the value $v_i$ for the all-pay auctions. \Cref{lem:valueLambda} establishes the uniqueness of the Lagrange multipliers, hence the uniqueness of the value $v_i$.

A potential issue that arises is whether the constraint that the strategy $P_i$ must be in $\mathfrak{B}_i$ leads to equilibria outside those characterized by \Cref{thm:uniqueCDFACB113}.
From Sklar's Theorem (\Cref{thm:Sklar}), we know that the joint distribution $P_i$ is equivalent to a set of marginal distributions $\left\{F_{i}^ j\right\}_{j=1}^3$, together with a suitable $3$-copula $C$. So if a suitable $3$-copula exists, the constraint that $P_i$ be in $\mathfrak{B}_i$ places no restraint on the set of potential univariate marginal distribution functions, $\left\{F_{i}^ j\right\}_{j=1}^3$; instead, this constraint and the set of univariate marginal distributions places a restraint on the set of feasible $3$-copulas. Since \Cref{thm:existCopula} establishes the existence of suitable $3$-copula, this is not an issue.

On the other hand, the restriction on the $3$-copula implies that the set of equilibrium $3$-variate distributions for the game forms a strict subset of the set of all $3$-variate distribution functions with univariate marginal distribution functions characterized by \Cref{thm:uniqueCDFACB113}.

The proof of \Cref{thm:uniqueCDFACB113} under the assumption that suitable $3$-copula exists is contained in the results that fill up the rest of this section. The proof takes inspiration from the proofs found in \cite{Roberson2006} and \cite{BayeEtAl}.

First, for the form of the Lagrangian in \Cref{prop:Lagrange} to be accurate, we need show that there are no atoms in any Nash equilibrium strategies. The following theorem proves this in the more general case of equal levels of force for both players and any $n$ number of battlefields where $n\ge3$:
\begin{theorem}
  \label{thm:noAtom}
  If $n\ge 3$, then Nash equilibrium strategies for $ACB(1,1,n)$ cannot contain atoms. 
\end{theorem}
\begin{proof}
  Suppose that we have an equilibrium strategy $P_A$ with an atom in battlefield $j$ on $a_j$. 

  Let $p=(b_1,b_2,\dots,b_{j-1},b_j=a_j,b_{j+1},\dots,b_n)$ be any pure strategy in the support of $P_A$ that contains playing $a_j$ on battlefield $j$. 
  The general idea of this proof will be to find a pure strategy $p'$ that does strictly better against $P_A$ than $p$, thus reaching a contradiction that $P_A$ cannot be an equilibrium strategy as we supposed.

  Let $f_k(a)$ denote the possibility of choosing $a$ on battlefield $k$ in $P_A$.
  If $b_k$ is any point that is not an atom and $b_k$ is greater than $b_{k-1}$ (or greater than $0$ in the case of $b_1$), then consider the pure strategy $p'$ that plays $\epsilon$ lower on battlefield $k$ and plays $\epsilon' = \frac{\epsilon}{n-j+1}$ higher on battlefield $j$ and all the battlefields after that. 
  We can always find sufficiently small positive $\epsilon$ and $\epsilon'$ such that there is no atom between $b_k$ and $b_k-\epsilon$ in $P_A$. So the payoff of $p'$ against $P_A$ minus the payoff of $p$ against $P_A$ is at least 
  \[
  \frac{1}{n}f_k(a_j) - \delta
  \]
  for any $\delta>0$.
  Hence, $p'$ does strictly better than $p$ against $P_A$.

  Therefore, for $P_A$ to be an equilibrium strategy, all such $b_k$ that are not atoms must be equal to $b_{k-1}$ (or $0$ if $k=1$). So every pure strategy $p$ in the support of $P_A$ containing the atom $a_j$ on battlefield $j$ must be of the following form:
  a series of zeros in the first few battlefields (possibly none), an atom, the same level of force in the next few battlefields (also possibly none), another atom, the same level of force (as in the previous atom) in the next few battlefields, and so forth.

  Now, one of the following statements must be true:
  \begin{enumerate}
    \item 
      All $p$ in the support of $P_A$ containing the atom $a_j$ on battlefield $j$ is played with probability $0$.
    \item
      There exists some $p=(c_1,\dots,c_n)$ in the support of $P_A$ containing the atom $a_j$ on battlefield $j$ that is played with a positive probability, hence every $c_k$ is an atom on battlefield $k$.
  \end{enumerate}

  Suppose that statement $1$ is true. For $a_j$ to be played with some positive probability, there must be a continuum of such $p$. 
  Hence there must also be a continuum of atoms, which is clearly impossible. 
  So statement $2$ must be true. 
  Let $q = (c_1,\dots,c_n)$ be such a pure strategy in the support of $P_A$ where every $c_k$ is an atom on battlefield $k$. 
  Some casework is needed here:

  \begin{enumerate}
    \item 
      All the $c_k$ are the same. Then they must all be $\frac{1}{n}$. 
      In any pure strategy where player $A$ plays $\frac{1}{n}$ on the first battlefield, he must also play $\frac{1}{n}$ on all the other battlefields. 
      So $f_1\left(\frac{1}{n}\right)\le f_k\left(\frac{1}{n}\right)$ for all $k\ge 2$. 
      Hence,
      \[
      f_1\left(\frac{1}{n}\right) < \sum_{k=2}^n f_k\left(\frac{1}{n}\right).
      \]
      Consider the pure strategy $q'$ that plays $\left( \frac{1}{n}-\epsilon \right)$ on battlefield $1$ and plays $\left( \frac{1}{n}+\frac{\epsilon}{n-1} \right)$ on all the other battlefields. We can find a sufficiently small positive $\epsilon$ such that there are no atoms between $\frac{1}{n}$ and $\left( \frac{1}{n}-\epsilon \right)$ on battlefield $1$. The payoff of $q'$ against $P_A$ minus the payoff of $q$ against $P_A$ is at least 
      \[
      \frac{1}{n}\cdot \left( \sum_{k=2}^n f_k\left(\frac{1}{n}\right)- f_1\left(\frac{1}{n}\right)  \right) - \delta
      \]
      for any $\delta>0$.
      So $q'$ does strictly better against $P_A$ than $q$.

    \item
      All the $c_k$ fall into exactly two values, $d_0$ and $d_1$. ($d_0<d_1$) Suppose $q$ contains $m$ battlefields with level of force $d_0$ and then $(n-m)$ battlefields with level of force $d_1$.i
      \begin{enumerate}
        \item 
          If $d_0=0$, then $d_1 = \frac{1}{n-m}$. 
          Given any pure strategy in the support of $P_A$ that plays $d_1$ on battlefield $(m+1)$, it must also play $d_1$ on all the battlefields after that, and play $0$ on the battlefields $1$ to $m$. So $f_{m+1}(d_1)\le f_k(c_k)$ where $k\neq m+1$. 
          Hence, 
          \[
          \sum_{k\neq m+1}f_k(c_k)>f_{m+1}(d_1=c_{m+1}).
          \]
          Consider the pure strategy $q'$ that plays $d_1-\epsilon$ on battlefield $(m+1)$ and plays $c_k+\frac{\epsilon}{n-1}$ on battlefield $k$ for all $k\neq m+1$.
          We can find a sufficiently small positive $\epsilon$ such that there are no atoms between $d_1$ and $d_1-\epsilon$ on battlefield $(m+1)$. The payoff of $q'$ against $P_A$ minus the payoff of $q$ against $P_A$ is at least:
          \[
          \frac{1}{n}\cdot\left(  \sum_{k\neq m+1}f_k(c_k) - f_{m+1}(d_1)  \right) 
          - \delta
          \]
          for any $\delta>0$.
          So $q'$ does strictly better against $P_A$ than $q$.
        \item
          If $d_0>0$, then at least one of the following two must be true:
          \label{twoValues}
          \begin{enumerate}
            \item 
              \[
              f_{m+1}(c_{m+1})< \sum_{k\neq m+1}f_k(c_k).
              \]
            \item
              \[
              \sum_{k=m+1}^nf_k(c_k)>\sum_{k=1}^mf_k(c_k).
              \]
          \end{enumerate}
          Similar to the arguments above, if the first one is true, then we can construct a $q'$ by playing $\epsilon$ lower on battlefield $(m+1)$ and $\epsilon'$ higher on all the other battlefields; if the second one is true, then we can construct a $q'$ by playing $\epsilon$ higher on battlefield $(m+1)$ and all the battlefields after that, and playing $\epsilon'$ lower on battlefields $1$ to $m$. 
          In either case, $q'$ does strictly better than $q$ against $P_A$.
      \end{enumerate}
    \item 
      All the $c_k$ fall into at least three different values.
      So from these values we can choose two different values that are not zero. Then we apply the proof in \cref{twoValues} and obtain the needed pure strategy $q'$.
  \end{enumerate}
  In all the cases, a contradiction is reached, showing that $P_A$ cannot be an equilibrium strategy.
\end{proof}


In the following discussions, let $P = \left\{ F^j \right\}_{j=1}^3$ be any joint distribution characterized in \Cref{thm:uniqueCDFACB113}, and let $P' = \left\{ f^j \right\}_{j=1}^3$ be any equilibrium strategy. Our goal is to prove that $P$ is an equilibrium strategy, and that $P$ and $P'$ have the same marginal distributions.
\begin{lemma}
  Suppose $p = (a,b,c)$ is any pure strategy in $\mathfrak{B}_A=\mathfrak{B}_B=\mathfrak{B}$. Then the payoff of $p$ against $P$ is $\frac{1}{2}$ if $0\le a\le\frac{1}{3}$, $\frac{1}{6}\le b\le\frac{1}{2}$, and $\frac{1}{3}\le c\le \frac{2}{3}$; and the payoff is less than $\frac{1}{2}$ otherwise.
  \label{lem:payoffPureAgainstP}
\end{lemma}
\begin{proof}

  Suppose $A$ plays the mixed strategy $P$ and $B$ plays the pure strategy $p=(a,b,c)$, where $0\leq a\leq b\leq c$ and $a+b+c=1$. Then, let $W(a,b,c)$ be the payoff for $B$. So
  \[
  W(a,b,c)=\frac{1}{3}\left(F^1(a)+F^2(b)+F^3(c)\right)
  \]
  Our goal is to find the maximum value of $W(a,b,c)$ in $\mathfrak{B}$ and to show that it is no greater than $0$. 

  Clearly, $0\leq a\leq \frac{a+b+c}{3} = \frac{1}{3}$, so $F^1(a)=3a$. And $b\leq \frac{b+c}{2}\leq\frac{1}{2}$
  \begin{itemize}
    \item If $b<\frac{1}{6}$, then $c=1-a-b\geq 1-2b>\frac{2}{3}$, so $F^2(b)=0$ and $F^3(c)=1$. And $a\leq b<\frac{1}{6}$
      \begin{align*}
        W(a,b,c)&=\frac{1}{3}(3a+0+1)\\
        &< \frac{1}{3}(3\cdot\frac{1}{6}+0+1)\\
        &=\frac{1}{2}.
      \end{align*}
      So $W(a,b,c)<\frac{1}{2}$.
    \item If $b\geq\frac{1}{6}$, then $F^2(b)=-\frac{1}{2}+3b$. Since $c\geq\frac{1}{3}$, $F^3(c)\leq 3c-1$.
      \begin{align*}
        W(a,b,c)&\leq\frac{1}{3}(3a-\frac{1}{2}+3b+3c-1)\\
        &= (a+b+c)-\frac{1}{2}\\
        &=\frac{1}{2}.
      \end{align*}
      Equality holds if and only if $F^3(c) = 3c-1$, which is equivalent to $\frac{1}{3}\le c\le \frac{2}{3}$. In this case $0\le a\le\frac{1}{3}$, $\frac{1}{6}\le b\le\frac{1}{2}$, and $\frac{1}{3}\le c\le \frac{2}{3}$. 

      Otherwise, equality does not hold, and the payoff is less than $\frac{1}{2}$.
  \end{itemize}
\end{proof}

\begin{lemma}
  Any joint strategy $P$ as characterized in \Cref{thm:uniqueCDFACB113} is a Nash equilibrium strategy.
  \label{lem:isNE}
\end{lemma}
\begin{proof}
  We know that the game $ACB(1,1,3)$ is symmetrical and has constant sum $1$, and since \Cref{lem:payoffPureAgainstP} indicates that $P$ gives a payoff of at least $\frac{1}{2}$ against any pure strategy, so $P$ must be an equilibrium strategy.
\end{proof}

Let $\bar{s}^1 = \frac{1}{3}$, $\underline{s}^1 = 0$, $\bar{s}^2 = \frac{1}{2}$, $\underline{s}^2 = \frac{1}{6}$, $\bar{s}^3 = \frac{1}{3}$, $\underline{s}^3 = \frac{2}{3}$.
Clearly, $\bar{s}^j$ is just the upper bound of $P$ on battlefield $j$, and $\underline{s}^j$ is the lower bound.
\begin{lemma}
  $F^j(x^j) = \frac{x^j-\underline{s}^j}{\bar{s}^j - \underline{s}^j}$ for $\underline{s}^j\le x^j\le \bar{s}^j$ and all $j$. 
  \label{lem:FJ}
\end{lemma}

\begin{proof}
  This is self-evident from the representation of $F^j$ in \Cref{thm:uniqueCDFACB113}:
  \begin{align*}
    F^1(u)          & =                    
    \begin{cases}
      3u              & 0\leq u\leq\frac{1}{3}           \\
      1               & \frac{1}{3}<u\leq 1              
    \end{cases}
    \\
    F^2(u)          & =                      
    \begin{cases}
      0               & 0\leq u<\frac{1}{6}              \\	
      -\frac{1}{2}+3u & \frac{1}{6}\leq u\leq\frac{1}{2} \\
      1               & \frac{1}{2}<u\leq 1              
    \end{cases}
    \\
    F^3(u)          & =                       
    \begin{cases}
      0               & 0\leq u<\frac{1}{3}              \\	
      -1+3u           & \frac{1}{3}\leq u\leq\frac{2}{3} \\
      1               & \frac{2}{3}<u\leq 1.              
    \end{cases}
  \end{align*}
\end{proof}

\begin{lemma}
  If $x < \underline{s}^j$, then $f^j(x) = 0$. If $x > \bar{s}^j$, then $f^j(x) = 1$. Or, in other words, $P'$ does not place any strategy outside $\left[ \underline{s}^j, \bar{s}^j \right]$.
  \label{lem:supportP}
\end{lemma}

\begin{proof}
  Since $ACB(1,1,3)$ is a two player symmetric constant sum $1$ game, every pure strategy in the support of $P'$, an equilibrium strategy, must give the unique equilibrium payoff, $\frac{1}{2}$, when played against another equilibrium strategy, $P$. From \Cref{lem:payoffPureAgainstP} we know that a pure strategy $p$ only gives payoff $\frac{1}{2}$ against $P$ when $p$ plays a level of force between $\underline{s}^j$ and $\bar{s}^j$ on battlefield $j$ for all $j$.
  So $P'$ cannot play any strategy outside that range.
\end{proof}

\begin{corollary}
  $f^j(\underline{s}^j) = 0$ and $f^j(\bar{s}^j) = 1$.
  \label{endPointsP}
\end{corollary}

\begin{proof}
  \Cref{thm:noAtom} implies that $f^j$ is continuous. This, together with \Cref{lem:supportP}, gives the desired result.
\end{proof}

Let us recall player $i$'s optimization problem for $ACB(1,1,3)$ (\cref{eq:Lagrange} in \Cref{prop:Lagrange}):
\begin{equation}
  \max\limits_{\left\{F_{i}^ j\right\}_{j=1}^3}\lambda_i\sum_{j=1}^3\left[\int_0^\infty \left[ \frac{1}{3\lambda_i}F^j_{-i}(x)-x \right]dF^j_i \right]+\lambda_iX_i
  \label{eq:Lagrange113}
\end{equation}
where the set of univariate marginal distribution functions $\left\{F_{i}^ j\right\}_{j=1}^3$ satisfy the constraint that there exists a $3$-copula $C$ such that the support of the $3$-variate distribution 
\[
C\left( F^1_i\left(x^1\right),F^2_i\left(x^2\right),F^3_i\left(x^3\right) \right)
\] 
is contained in $\mathfrak{B}_i$. 

From the lemmas above, we can add some further restrictions to it. From \Cref{lem:supportP}, we know that $P_i$ must be played within $\left[ \underline{s^j},\overline{s^j} \right]$ for every battlefield $j$. From \Cref{lem:isNE}, we know that $P$ is an equilibrium strategy, so $P_i$ must be a best response against $P$ and vice versa. 
Since \Cref{thm:existCopula} establishes the existence of suitable $3$-copula, we can disregard that restriction for now and focus on the rest.

For different $j$, $F_i^j$ is independent. So \cref{eq:Lagrange113} is the maximization of three independent sums, hence the sum of three independent maximizations:
\begin{multline*}
  \max\limits_{\left\{F_{i}^ j\right\}_{j=1}^3}\lambda_i\sum_{j=1}^3\left[\int_0^\infty \left[ \frac{1}{3\lambda_i}F^j_{-i}(x)-x \right]dF^j_i \right]+\lambda_iX_i \\
  = \sum_{j=1}^3 \max\limits_{F_{i}^ j}\lambda_i\int_0^\infty \left[ \frac{1}{3\lambda_i}F^j_{-i}(x)-x \right]dF^j_i +\lambda_iX_i.
\end{multline*}
The term $\lambda_iX_i$ is just a constant, so we could throw that away. Thus the problem for player $i$ becomes:
\[
\max\limits_{F_{i}^ j}\lambda_i\int_0^\infty \left[ \frac{1}{3\lambda_i}F^j_{-i}(x)-x \right]dF^j_i 
\]
for all battlefields $j$, under the constraint that $P_A$ is a best response against $P$, $P$ is a best response against $P_A$, and $P_A$ is played within $\left[ \underline{s}^j,\bar{s}^j \right]$. Let us set $P_A = P' = \left\{ f^j \right\}_{j=1}^3$. 
Since we assume the existence of a suitable $3$-copula, the different $f^j$ can be considered independent and the different maximizations for different battlefields can also be considered independent.
Hence, $f^j$ and $F^j$ form an equilibrium for all $j$.

Let $B_i^j(x_i^j,F^j_{-i}) = \lambda_i\left( \frac{1}{3\lambda}F_{-i}^j\left( x^j_i \right) -x_i^j \right)$. This is the payoff for player $i$ by playing $x_i^j$ when player $-i$ plays $F^j_{-i}$ in the maximization problem for battlefield $j$.

\begin{lemma}
  $B^j_i(x^j,f^j) = \lambda_i\left( \frac{1}{3\lambda_i}f^j\left( x^j \right) -x^j \right)$ is constant for all $\underline{s}^j\le x^j\le \bar{s}^j$.
  \label{lem:constantB}
\end{lemma}
\begin{proof}
  Since $F^j$ is an equilibrium strategy against $f^j$, every strategy in the support of $F$ gives a constant payoff against $f^j$. Since the support of $F^j$ is $\left[ \underline{s}^j , \bar{s}^j \right]$, the result directly follows.
\end{proof}

\begin{lemma}
  $B^j_i(x^j,f^j) = \lambda_i\left( \frac{1}{3\lambda_i}f^j\left( x^j \right) -x^j \right) = -\lambda_i\underline{s}^j = \frac{1}{3}-\lambda_i\bar{s}^j$ for all $\underline{s}^j\le x^j\le \bar{s}^j$.
  \label{lem:valueB}
\end{lemma}

\begin{proof}
  From \Cref{endPointsP}, $B^j_i(\underline{s}^j,f^j) = -\lambda_i\underline{s}^j$, and $B^j_i(\bar{s}^j,f^j) =  \frac{1}{3}-\lambda_i\bar{s}^j $. The result directly follows from \Cref{lem:constantB}.
\end{proof}

\begin{lemma}
  $\lambda_i = 1$ for all $i$.
  \label{lem:valueLambda}
\end{lemma}
\begin{proof}
  From \Cref{lem:valueB}, we have $\lambda_i = \frac{1}{3\left(\bar{s}^j - \underline{s}^j\right)}$. Note that $\left(\bar{s}^j - \underline{s}^j\right)$ is always $\frac{1}{3}$ for all $j$, so $\lambda_i = 1$.
\end{proof}

\begin{lemma}
  $f^j(x^j) = F^j(x^j)$ for all $j$ and all $x^j$.
  \label{lem:uniqueMarginalDistribution}
\end{lemma}

\begin{proof}
  From \Cref{lem:valueB,lem:valueLambda}, we have $f^j(x^j) = \frac{x^j-\underline{s}^j}{\bar{s}^j - \underline{s}^j}$ for $\underline{s}^j\le x^j\le \bar{s}^j$ and all $j$.
  From \Cref{lem:FJ}, the value of $f^j$ coincides with the value of $F^j$ here.
  \Cref{endPointsP} ensures that $f^j$ and $F^j$ are the same elsewhere.
\end{proof}

With these lemmas, we can prove the uniqueness of the marginal distributions in the Nash equilibria of the game $ACB(1,1,3)$. We restate the theorem here for convenience.
\uniqueCDFACB*
\begin{proof}[Proof of~\Cref{thm:uniqueCDFACB113}]
  From \Cref{lem:isNE} we know that every joint distribution with marginal distribution functions as characterized above is a Nash equilibrium strategy, hence the second part of the theorem is proved.

  \Cref{lem:uniqueMarginalDistribution} establishes the uniqueness of marginal distributions of Nash equilibrium strategies, and proves that these marginal distributions are exactly those characterized above. Hence, we have proven the first part of the theorem.
\end{proof}

\section{Unique equilibrium payoffs of the game $ACB(X_A,X_B,2)$}
\label{sec:w2t}
In this section we find the unique equilibrium payoffs of all cases of the Asymmetric Colonel Blotto game involving only two battlefields. 

Suppose without loss of generality that $X_A=1$ and $X_B=t\leq1$.

Let $W_n(t)$ denote the payoff for $A$ in a Nash equilibrium in such a game with $n$ battlefields. From \Cref{thm:uniqueNEPayoff}, we know that $W_n(t)$ is well-defined.
\begin{theorem}
  \label{thm:w2t}
  $$W_2(t)=
  \begin{array}{ccc}
    \frac{k+2}{2k+2}, & \frac{2k}{2k+1}\leq t< \frac{2k+2}{2k+3} & $where $ k=0,1,2,\dots\\
  \end{array}
  $$
  See \Cref{fig:w2t} for a graphical representation.
\end{theorem}

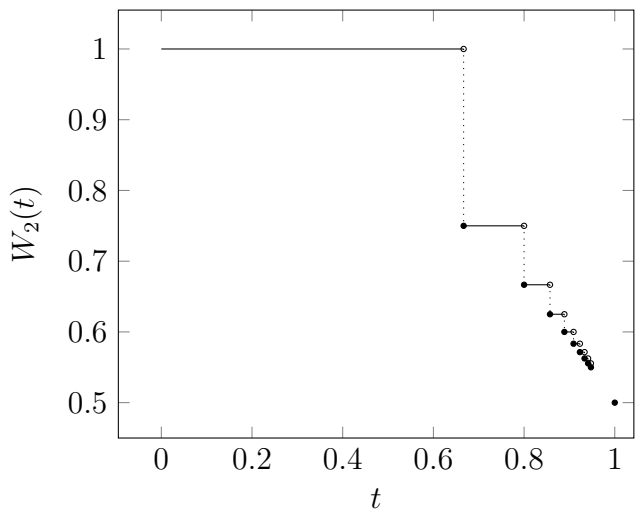
\begin{figure}[thb]

  \centering
  \begin{tikzpicture}
    \begin{axis}[
      xlabel = {$t$},
      ylabel = {$W_2(t)$},
      ymin = 0.45,
      ytick = {0.5,0.6,...,1.1}
      ]
      \pgfplotsinvokeforeach {0,1,...,8}{
      \addplot[domain={(2*#1)/((2*#1+1))}:{(2*#1+2)/(2*#1+3)},black] {(#1+2)/(2*#1+2)};
      \draw (axis cs:{(2*#1+2)/(2*#1+3)},{(#1+2)/(2*#1+2)}) circle [radius = 1pt];
      \draw[dotted] (axis cs:{(2*#1+2)/(2*#1+3)},{(#1+2)/(2*#1+2)}) -- (axis cs:{(2*#1+2)/(2*#1+3)},{(#1+3)/(2*#1+4)});
      \draw[fill=black] (axis cs:{(2*#1+2)/(2*#1+3)},{(#1+3)/(2*#1+4)}) circle (1pt);
      }
      \draw[fill=black] (axis cs:1,0.5) circle (1pt);
    \end{axis}
  \end{tikzpicture}
  \caption{$W_2(t)$, the payoff for Asymmetric Colonel Blotto Game with $2$ battlefields}
  \label{fig:w2t}
\end{figure}

The proof for $W_2(t)$ and constructions of Nash equilibriums can be found later in this section. Before we go on to prove this, let us first prove a lemma regarding the Asymmetric Colonel Blotto game with two battlefields:
\begin{lemma}
  \label{lem:pureAgainstPureNEq2}
  Suppose that $X_A>X_B$. If player $A$ deploys the pure strategy $(a,X_A-a)$ and $B$ deploys the pure strategy $(b,X_B-b)$, then 
  $$W_A = \begin{cases}
    \frac{1}{2}           & a-b<0 $ or $ a-b>X_A-X_B \\
    \frac{3}{4} & a-b=0 $ or $ a-b=X_A-X_B \\
    1           & 0<a-b<X_A-X_B            \end{cases}
  $$
  where $W_A$ is the payoff for player $A$.
\end{lemma}
\begin{proof}
  \leavevmode
  \begin{itemize}
    \item If $a-b<0$, then $X_A-a>X_B-b$, so $W_A=\frac{1}{2}$.
    \item If $a-b=0$, then $X_A-a>X_B-b$, so $W_A=\frac{3}{4}$.
    \item If $0<a-b<X_A-X_B$, then $X_A-a>X_B-b$, so $W_A=1$.
    \item If $a-b=X_A-X_B>0$, then $W_A=\frac{3}{4}$.
    \item If $a-b>X_A-X_B$, then $X_A-a<X_B-b$, so $W_A=\frac{1}{2}$.
  \end{itemize}
  Hence,
  $$W_A = \begin{cases}
    \frac{1}{2}           & a-b<0 $ or $ a-b>X_A-X_B \\
    \frac{3}{4} & a-b=0 $ or $ a-b=X_A-X_B \\
    1           & 0<a-b<X_A-X_B. \end{cases}
  $$
\end{proof}

With the help of \Cref{lem:pureAgainstPureNEq2}, we can prove \Cref{thm:w2t}:
\begin{proof}[Proof of~\Cref{thm:w2t}]
  \leavevmode
  \begin{enumerate}
    \item Suppose that $t<\frac{2}{3}$.
      In this case player $A$ can simply overwhelm player $B$ in all the battlefields. 
      Take $P_A=\left(\left(\frac{1}{3},\frac{2}{3}\right),1\right)$ and $P_B$ to be any strategy. $P_A$ and $P_B$ form a Nash equilibrium and $W_2(t) = 1$.

      Given any pure strategy $(x,t-x)$ of $B$, we must have $x\leq t-x$, so $x\leq \frac{t}{2}$, $\frac{1}{3}>\frac{t}{2}\geq x$, and $\frac{2}{3}>t\geq t-x$. 
      Thus in this case, the payoff to $B$ is $0$. This means that $B$ cannot increase payoff regardless of the strategy (or mixed strategy) chosen. On the other hand, the payoff of $A$ is $1$, which is the maximum possible value, so clearly neither can $A$ increase payoff by changing strategy. Hence,  $P_A=\left(\left(\frac{1}{3},\frac{2}{3}\right),1\right)$ and any $P_B$ form a Nash equilibrium and $W_2(t) = 1$.
    \item Suppose $k$ is such that $\frac{2k}{2k+1}\leq t< \frac{2k+2}{2k+3}$, where $k\in\mathbb{Z}^+$.
      Take
      $$P_A=\left\{\left(\left(\epsilon+j(1-t),1-\epsilon-j(1-t)\right),\frac{1}{k+1}\right)\,\middle|\,0\leq j\leq k\right\}$$
      and 
      $$P_B=\left\{\left(\left(j(1-t),t-j(1-t)\right),\frac{1}{k+1}\right)\,\middle|\,0\leq j\leq k\right\},$$
      where $\epsilon$ is such that
      \begin{equation}
        \label{eq:w2tEpsilon}
        \frac{2k+1}{2} t-k<\epsilon<\min\left(1-t,t k-k+\frac{1}{2}\right).
      \end{equation}
      The notation here just means that player $A$ plays pure strategy 
      \[\left(\epsilon+j(1-t),1-\epsilon-j(1-t)\right)\] 
      with probability $\frac{1}{k+1}$ for all $j$ such that $0\leq j\leq k$; and player $B$ plays pure strategy
      \[\left(j(1-t),t-j(1-t)\right)\]
      with probability $\frac{1}{k+1}$ for all $j$ such that $0\leq j\leq k$.
      Then we claim that $P_A$ and $P_B$ form a Nash equilibrium and $W_2(t)=\frac{k+2}{2k+2}$.

      First we will show that these mixed strategies are legitimate.
      If $t<\frac{2k+2}{2k+3}$, then $\frac{2k+3}{2} t<k+1$, so $\frac{2k+1}{2} t-k<1-t$.
      Now $\frac{t}{2}<\frac{1}{2}$, so $\frac{2k+1}{2}t-k<tk-k+\frac{1}{2}$,
      which in turn implies that 
      $$\frac{2k+1}{2}\cdot t-k\geq k-k=0.$$
      Hence, a positive $\epsilon$ satisfying \cref{eq:w2tEpsilon} exists. 
      Further, we need to check that the level of force distributed on the first battlefield, $x_1$, is less than or equal to the force distributed on the second battlefield, $x_2$; or, equivalently, for player $i$, we need to check that $x_1\le \frac{X_i}{2}$.
      First, let's check player $A$'s strategy. Since $j\le k$, we must have
      \[\epsilon+j(1-t)\leq \epsilon+k(1-t).\]
      Then we plug in the upper bound of $\epsilon$ in \cref{eq:w2tEpsilon} and get
      \[\epsilon+k(1-t)<t\cdot k-k+\frac{1}{2}+k(1-t)=\frac{1}{2}.\]
      So $\epsilon+j(1-t)<\frac{1}{2}$.
      Now let's check player $B$'s strategy.
      We already know that $t\ge\frac{2k}{2k+1}$, rearrange and we would get
      \[ k(1-t)\leq\frac{t}{2} .\]
      Since $ j(1-t)\leq k(1-t)$, we must have
      \[j(1-t)\leq\frac{t}{2} .\] 
      So both $P_A$ and $P_B$ are legitimate mixed strategies.

      Suppose $A$ chooses some pure strategy $p'_A=(x,1-x)$. Set $a=\left\lfloor\frac{x}{1-t}\right\rfloor$.
      Hence, $$(1-t)a\leq x<(1-t)(a+1)$$
      where $0\leq a \leq k+1$.
      Now let us expand $P_B$ into pure strategies in the calculation of $W_A\left( p'_A, P_B \right)$:
      \begin{align*}
        (k+1)W_A(p'_A,P_B)&=\sum_{j=0}^kW_A\left((x,1-x),\left(j(1-t),t-j(1-t)\right)\right)\\
        &\leq \sum_{j=0}^{a-2}W_A\left((x,1-x),\left(j(1-t),t-j(1-t)\right)\right)\\
        &\qquad +\sum_{j=a+1}^kW_A\left((x,1-x),\left(j(1-t),t-j(1-t)\right)\right)\\
        &\qquad +W_A\left((x,1-x),\left((a-1)(1-t),t-(a-1)(1-t)\right)\right)\\
        &\qquad +W_A\left((x,1-x),\left(a(1-t),t-a(1-t)\right)\right).
      \end{align*}
      There is a $\leq$ sign on the second line since if $a=k+1$, there is one additional non-negative term on the right, 
      $W_A\left((x,1-x),\left((k+1)(1-t),t-(k+1)(1-t)\right)\right)$, compared with the original formula.

      First let us consider the sum \[\sum_{j=0}^{a-2}W_A\left((x,1-x),\left(j(1-t),t-j(1-t)\right)\right).\] Here,
      \begin{align*}
        x-j(1-t)&\geq x-(a-2)(1-t)\\
        &\geq a(1-t)-(a-2)(1-t)\\
        &>(1-t).
      \end{align*}
      Hence, according to \Cref{lem:pureAgainstPureNEq2}, \[W_A\left((x,1-x),\left(j(1-t),t-j(1-t)\right)\right)=\frac{1}{2}\] if $0\leq j\leq a-2$. Thus, 
      $$\sum_{j=0}^{a-2}W_A\left((x,1-x),\left(j(1-t),t-j(1-t)\right)\right)=\frac{a-1}{2}.$$   

      Then let us consider the sum \[\sum_{j=a+1}^kW_A\left((x,1-x),\left(j(1-t),t-j(1-t)\right)\right).\] Here,
      \begin{align*}
        x-j(1-t)&\leq x-(a+1)(1-t)\\
        &< (a+1)(1-t)-(a+1)(1-t)\\
        &<0.
      \end{align*}	
      Hence, according to \Cref{lem:pureAgainstPureNEq2}, \[W_A\left((x,1-x),\left(j(1-t),t-j(1-t)\right)\right)=\frac{1}{2}\] if $a+1\leq j\leq k$. Thus,
      $$\sum_{j=a+1}^kW_A\left((x,1-x),\left(j(1-t),t-j(1-t)\right)\right)=\frac{k-a}{2}.$$

      If $x=a(1-t)$, then according to \Cref{lem:pureAgainstPureNEq2},
      \begin{multline*}
        W_A\left((x,1-x),\left((a-1)(1-t),t-(a-1)(1-t)\right)\right)\\
        +W_A\left((x,1-x),\left(a(1-t),t-a(1-t)\right)\right)
        =\frac{3}{4}+\frac{3}{4}=\frac{3}{2}.
      \end{multline*}

      If $x>a(1-t)$, then according to \Cref{lem:pureAgainstPureNEq2},
      \begin{multline*}
        W_A\left((x,1-x),\left((a-1)(1-t),t-(a-1)(1-t)\right)\right)\\
        +W_A\left((x,1-x),\left(a(1-t),t-a(1-t)\right)\right)=\frac{1}{2}+1=\frac{3}{2}.
      \end{multline*}

      Hence, $W_A(p'_A,P_B)\leq \frac{1}{k+1} \cdot \frac{k+2}{2} = \frac{k+2}{2k+2}$.
      So $A$ cannot increase payoff above $\frac{k+2}{2k+2}$ by changing strategy.
      Now let's consider player $B$'s strategy.
      Suppose $B$ chooses some pure strategy $p'_B=(y,t-y)$. Set $b=\lceil\frac{y-\epsilon}{1-t}\rceil$.
      Hence, $$(1-t)(b-1)+\epsilon< y\leq (1-t)b+\epsilon$$
      where $0\leq b\leq k$.

      \begin{enumerate}
        \item
          In the case where $0\leq b\leq k-1$,
          let's expand $P_A$ into pure strategies in the calculations of $W_A\left( P_A, p'_B \right)$:
          \begin{align*}
            (k+1)W_A(P_A,p'_B)&=\sum_{j=0}^kW_A\left((j(1-t)+\epsilon,1-j(1-t)-\epsilon),\left(y,t-y\right)\right)\\
            &= \sum_{j=0}^{b-1}W_A\left((j(1-t)+\epsilon,1-j(1-t)-\epsilon),\left(y,t-y\right)\right)\\
            &\qquad +\sum_{j=b+2}^kW_A\left((j(1-t)+\epsilon,1-j(1-t)-\epsilon),\left(y,t-y\right)\right)\\
            &\qquad +W_A\left((b(1-t)+\epsilon,1-b(1-t)-\epsilon),\left(y,t-y\right)\right)\\
            &\qquad +W_A\left(((b+1)(1-t)+\epsilon,1-(b+1)(1-t)-\epsilon),\left(y,t-y\right)\right).
          \end{align*}

          First let us consider the sum \[\sum_{j=0}^{b-1}W_A\left((j(1-t)+\epsilon,1-j(1-t)-\epsilon),\left(y,t-y\right)\right).\] Here,
          \[
            j(1-t)+\epsilon\leq (b-1)(1-t)+\epsilon
            <y.
          \]
          Hence, according to \Cref{lem:pureAgainstPureNEq2}, \[W_A\left((j(1-t)+\epsilon,1-j(1-t)-\epsilon),\left(y,t-y\right)\right)=\frac{1}{2}\] if $0\leq j\leq b-1$. Thus,
          $$\sum_{j=0}^{b-1}W_A\left((j(1-t)+\epsilon,1-j(1-t)-\epsilon),\left(y,t-y\right)\right)=\frac{b}{2}.$$   

          Then let us consider the sum \[\sum_{j=b+2}^kW_A\left((j(1-t)+\epsilon,1-j(1-t)-\epsilon),\left(y,t-y\right)\right).\] Here,
          \begin{align*}
            j(1-t)+\epsilon-y&\geq (b+2)(1-t)+\epsilon-y\\
            &\geq (b+2)(1-t)+\epsilon-b(1-t)-\epsilon\\
            &>1-t.
          \end{align*}	
          Hence, according to \Cref{lem:pureAgainstPureNEq2}, \[W_A\left((j(1-t)+\epsilon,1-j(1-t)-\epsilon),\left(y,t-y\right)\right)=\frac{1}{2}\] if $b+2\leq j\leq k$. Thus,
          $$\sum_{j=b+2}^kW_A\left((j(1-t)+\epsilon,1-j(1-t)-\epsilon),\left(y,t-y\right)\right)=\frac{k-b-1}{2}.$$

          If $y=b(1-t)+\epsilon$, then according to \Cref{lem:pureAgainstPureNEq2},
          \begin{multline*}
            W_A\left((b(1-t)+\epsilon,1-b(1-t)-\epsilon),\left(y,t-y\right)\right)\\
            +W_A\left(((b+1)(1-t)+\epsilon,1-(b+1)(1-t)-\epsilon),\left(y,t-y\right)\right)
            =\frac{3}{4}+\frac{3}{4} = \frac{3}{2}.
          \end{multline*}

          If $y<b(1-t)+\epsilon$, then according to \Cref{lem:pureAgainstPureNEq2},
          \begin{multline*}
            W_A\left((b(1-t)+\epsilon,1-b(1-t)-\epsilon),\left(y,t-y\right)\right)\\
            +W_A\left(((b+1)(1-t)+\epsilon,1-(b+1)(1-t)-\epsilon),\left(y,t-y\right)\right)
            = 1+ \frac{1}{2} = \frac{3}{2}.
          \end{multline*}
          Therefore, in either case, $W_A(P_A,p'_B) = \frac{k+2}{2k+2}$.
        \item
          If $b=k$, then
          \[
            k(1-t)+\epsilon > k(1-t)+\frac{2k+1}{2}t-k
            =\frac{1}{2}t
            \geq y.
          \]
          So
          \begin{align*}
            (k+1)W_A(P_A,p'_B)&=\sum_{j=0}^kW_A\left((j(1-t)+\epsilon,1-j(1-t)-\epsilon),\left(y,t-y\right)\right)\\
            &= \sum_{j=0}^{k-1}W_A\left((j(1-t)+\epsilon,1-j(1-t)-\epsilon),\left(y,t-y\right)\right)\\
            &\qquad+W_A\left((k(1-t)+\epsilon,1-k(1-t)-\epsilon),\left(y,t-y\right)\right).
          \end{align*}

          Similarly,\[\sum_{j=0}^{k-1}W_A\left((j(1-t)+\epsilon,1-j(1-t)-\epsilon),\left(y,t-y\right)\right)=\frac{k}{2}.\] 

          And since $(1-t)(k-1)+\epsilon< y< (1-t)k+\epsilon$,
          $$W_A\left((k(1-t)+\epsilon,1-k(1-t)-\epsilon),\left(y,t-y\right)\right)=1.$$
      \end{enumerate}
      Hence, $W_A(P_A,p'_B)\geq \frac{k+2}{2k+2}$, which means that $B$ cannot increase payoff above $\frac{k}{2k+2}$ by changing strategy.
      Hence, $P_A$ and $P_B$ form a Nash equilibrium, and the equilibrium payoff for $A$ is $W_2(t)=\frac{k+2}{2k+2}$.

    \item Finally, suppose that $t=1$. Take any mixed strategy $P_A$ and any mixed strategy $P_B$. Then they form a Nash equilibrium with $W_2(t)=\frac{1}{2}$.
      To see this, suppose $A$ plays the pure strategy $P'_A=(a,1-a)$ and $B$ plays the pure strategy $P'_B=(b,1-b)$.
      \begin{itemize}
        \item       If $a=b$, clearly $W_A(P'_A,P'_B)=\frac{1}{2}$.
        \item If $a<b$, then $1-a>1-b$, so $W_A(P'_A,P'_B)=\frac{1}{2}$. Similarly, if $a>b$, $W_A(P'_A,P'_B)=\frac{1}{2}$.
      \end{itemize}
      Hence, the payoff is $\frac{1}{2}$ regardless of the pure strategies that both players play. As a result, the payoff is also $\frac{1}{2}$ regardless of what mixed strategies that the two players play.
  \end{enumerate}
\end{proof} 	      	      	      	      	      

\section{Unique equilibrium payoffs of the game $ACB(X_A,X_B,3)$}
\label{sec:w3t}
In this section we find the unique equilibrium payoffs of some cases of the Asymmetric Colonel Blotto game involving three battlefields. The results that follow are ordered by ascending values of $t$.

Suppose without loss of generality that $X_A=1$ and $X_B=t\leq1$. The case where $t=1$ is already solved in \Cref{sec:uniqueCDFACB113}, and we have $W_3(1) = \frac{1}{2}$. In the following discussions, let the function $s(x)$ be defined as follows:
\[
s(x) = 
\begin{cases}
  0 & x<0 \\
  \frac{1}{2} & x = 0\\
  1 & x>0.
\end{cases}
\]

\begin{theorem}
  \label{thm:w3t6/11}
  In the case where $t<\frac{6}{11}$, $w_3(t) = 1$.
\end{theorem}

\begin{proof}
  In this case, player $A$ can simply overwhelm player $B$ in all the battlefields. 

  Take $P_A=\{\left((\frac{2}{11},\frac{3}{11},\frac{6}{11}),1\right)\}$ and $P_B$ to be any strategy. $P_A$ and $P_B$ form a Nash equilibrium and $W_3(t) = 1$.
\end{proof}

\begin{theorem}
  \label{thm:w3t6/11_18/31}
  In the case where $\frac{6}{11}\leq t<\frac{18}{31}$,
  take 
  \begin{multline*}
    P_A=\left\{
    \left(\left(\frac{t}{3}+\epsilon,\frac{t}{2}+\epsilon,1-\frac{5}{6}t-2\epsilon\right),\frac{1}{3}\right), \right.
    \left(\left(\frac{t}{3}+\epsilon,1-\frac{4}{3}t-2\epsilon,t+\epsilon\right),\frac{1}{3}\right),\\
    \left. \left(\left(1-\frac{3}{2}t-2\epsilon,\frac{t}{2}+\epsilon,t+\epsilon\right),\frac{1}{3}\right) \right\}
  \end{multline*}
  and
  $$P_B=\left\{
  \left(\left(0,0,t\right),\frac{1}{3}\right),
  \left(\left(0,\frac{t}{2},\frac{t}{2}\right),\frac{1}{3}\right),
  \left(\left(\frac{t}{3},\frac{t}{3},\frac{t}{3}\right),\frac{1}{3}\right)
  \right\}$$
  where $0<\epsilon<\frac{1}{2}\left(1-\frac{31}{18}t\right)$.

  $P_A$ and $P_B$ form a Nash equilibrium and $W_3(t)=\frac{8}{9}$.
\end{theorem}

\begin{proof}

  Since $t<\frac{18}{31}$, a real $\epsilon$ satisfying the necessary condition must exist. From the range of $t$ and $\epsilon$, we can check that the strategies of $A$ and $B$ are legitimate, or in other words, the levels of force of the battlefields are nondecreasing.

  Suppose $A$ plays the pure strategy $p'_A=(a,b,c)$ where $a+b+c=1$. Then,
  \begin{align*}
    9W_A\left(p'_A,P_B\right)&=3W_A\left(\left(a,b,c\right),\left(0,0,t\right)\right)\\
    &+3W_A\left(\left(a,b,c\right),\left(0,\frac{t}{2},\frac{t}{2}\right)\right)\\
    &+3W_A\left(\left(a,b,c\right),\left(\frac{t}{3},\frac{t}{3},\frac{t}{3}\right)\right)\\
    &=2s\left(a-0\right)+s\left(a-\frac{t}{3}\right)\\
    &+s\left(b-0\right)+s\left(b-\frac{t}{2}\right)+s\left(b-\frac{t}{3}\right)\\
    &+s\left(c-\frac{t}{3}\right)+s\left(c-\frac{t}{2}\right)+s\left(c-\frac{t}{3}\right).
  \end{align*}
  For $W_A(p'_A,P_B)$ to be more than $\frac{8}{9}$, none of the terms on the right can be $0$. Hence, we must have
  \[
    a\geq\frac{t}{3},\qquad
    b\geq\frac{t}{2},\qquad
    c\geq t.\]
  So $$1=a+b+c\geq\frac{11}{6}t.$$
  This is only possible when $t=\frac{11}{6}$. In this case, $p'_A$ must be $(\frac{t}{3},\frac{t}{2},t)$, so $W_A(p'_A,P_B)=\frac{5}{6}<\frac{8}{9}$.
  Hence, $W_A(p'_A,P_B)\leq\frac{8}{9}$ for all pure strategies $p'_A$.

  Suppose $B$ plays the pure strategy $p'_B=(d,e,f)$ where $d+e+f=t$. Then,
  \begin{align*}
    9W_A\left(P_A,p'_B\right)&=3W_A\left(\left(\frac{t}{3}+\epsilon,\frac{t}{2}+\epsilon,1-\frac{5}{6}t-2\epsilon\right),\left(d,e,f\right)\right)\\
    &+3W_A\left(\left(\frac{t}{3}+\epsilon,1-\frac{4}{3}t-2\epsilon,t+\epsilon\right),\left(d,e,f\right)\right)\\
    &+3W_A\left(\left(1-\frac{3}{2}t-2\epsilon,\frac{t}{2}+\epsilon,t+\epsilon\right),\left(d,e,f\right)\right).
  \end{align*}
  Remember that $d\leq\frac{t}{3}$, $e\leq\frac{t}{2}$, and $f\leq t$. So,
  \[
  9W_A\left(P_A,p'_B\right)=6+s\left(1-\frac{3}{2}t-2\epsilon-d\right)+s\left(1-\frac{4}{3}t-2\epsilon-e\right)+s\left(1-\frac{5}{6}t-2\epsilon-f\right).
  \]
  From $\frac{6}{11}\leq t<\frac{18}{31}$ and $0<\epsilon<\frac{1}{2}\left(1-\frac{31}{18}t\right)$, we can show that
  $$d\geq 1-\frac{3}{2}t-2\epsilon\Rightarrow e\leq\frac{t-d}{2}\leq\frac{5}{4}t+\epsilon-\frac{1}{2}<1-\frac{4}{3}t-2\epsilon,$$
  $$e\geq 1-\frac{4}{3}t-2\epsilon\Rightarrow f\leq t-e\leq \frac{7}{3}t-1+2\epsilon<1-\frac{5}{6}t-2\epsilon,$$
  $$\text{and } f\geq 1-\frac{5}{6}t-2\epsilon\Rightarrow d\leq t-2f\leq \frac{8}{3}t-2+4\epsilon<1-\frac{3}{2}t-2\epsilon.$$
  Hence, at least $2$ terms on the right must be $1$.  So $W_A(P_A,p'_B)\geq\frac{8}{9}$ for all pure strategies $p'_B$.

  To conclude, player $A$ cannot increase payoff above $\frac{8}{9}$ by changing  strategy, 
  and player $B$ can also not increase payoff above $\frac{1}{9}$ by changing strategy. 
  So $P_A,P_B$ form a Nash equilibrium and the equilibrium payoff $W_3(t)=\frac{8}{9}$.
\end{proof}

\begin{theorem}
  \label{thm:w3t3/5_30/47}
  When $\frac{3}{5}<t<\frac{30}{47}$,    	
  take	
  $$P_A=\left\{\left(\left(\frac{30-22t}{75},\frac{15-11t}{25},\frac{11t}{15}\right),\frac{1}{2}\right),\left(\left(\frac{2t}{15},\frac{13t}{30},1-\frac{17t}{30}\right),\frac{1}{2}\right)\right\}$$ 
  and 
  $$P_B=\left\{\left(\left(\frac{t}{3},\frac{t}{3},\frac{t}{3}\right),\frac{1}{2}\right),\left(\left(0,0,t\right),\frac{1}{2}\right)\right\}.$$
Then $P_A$ and $P_B$ form a Nash equilibrium and $W_3(t)=\frac{5}{6}$.
\end{theorem}

\begin{proof}

  Let's begin with checking that the strategies distributed among the three battlefields are non-decreasing. First, let's check player $A$'s first strategy.
  $x_A^1$ is obviously smaller than $x_A^2$ ($x_i^j$ means the level of force player $i$ distributes on battlefield $j$):
  $$\frac{30-22t}{75}\leq \frac{30-22t}{75}\cdot\frac{3}{2}=\frac{15-11t}{25}.$$
  Since $t>\frac{3}{5}$ and $\frac{3}{5}>\frac{45}{88}$, we must have $5<88t$.
  Rearrange and we would get 
  \[\frac{15-11t}{25}<\frac{11t}{15}.\]
  Now let's check player $A$'s second strategy. Obviously, ${x_A^1}'$ is smaller than ${x_A^2}'$:
  $$\frac{2t}{15}<\frac{13t}{30}.$$
  Furthermore, since $t<1$, we can rearrange and obtain
  $$\frac{13t}{30}<1-\frac{17t}{30}.$$
  Hence, $P_A$ is legitimate. And clearly $P_B$ is legitimate.

  If player $A$ plays the pure strategy $p'_A=(a,b,1-a-b)$. Then,
  \begin{multline*}
    6W_A(P'_A,P_B) = 3W_A\left( (a,b,1-a-b),\left(\frac{t}{3},\frac{t}{3},\frac{t}{3}\right)\right)+3W_A\left((a,b,1-a-b),\left(0,0,t\right) \right)\\
    =s\left(a-\frac{t}{3}\right)+s\left(b-\frac{t}{3}\right)+s\left(1-a-b-\frac{t}{3}\right)\\
    +s\left(a-0\right)+s\left(b-0\right)+s\left(1-a-b-t\right).
  \end{multline*}
  For $W_A(p'_A,P_B)$ to be greater than $\frac{5}{6}$, none of the six terms on the right can be $0$. Hence, we obtain the following inequalities:
  \[
  \begin{array}{cc}
    a\geq\frac{t}{3} & a\geq0\\
    b\geq\frac{t}{3} & b\geq0\\
    1-a-b\geq\frac{t}{3} & 1-a-b\geq t.
  \end{array}
  \]
  Add together $a\geq\frac{t}{3}$, $b\geq\frac{t}{3}$, $1-a-b\geq t$, and we get
  $$1\geq\frac{5}{3}\cdot t.$$
  Hence, $t\geq \frac{3}{5}$, contradicting our hypothesis on $t$.
  Therefore, we must have $W_A(p'_A,P_B)\leq\frac{5}{6}$ given all pure strategy $p'_A$.

  If player $B$ plays the pure strategy $p'_B=(c,d,t-c-d)$, then,
  \begin{multline*}
    6W_A\left(P_A,p'_B\right)=3W_A\left(\left(\frac{30-22t}{75},\frac{15-11t}{25},\frac{11t}{15}\right),\left(c,d,t-c-d\right)\right)\\
    +3W_A\left(\left(\frac{2t}{15},\frac{13t}{30},1-\frac{17t}{30}\right),\left(c,d,t-c-d\right)\right)\\
    =s\left(\frac{30-22t}{75}-c\right)+s\left(\frac{15-11t}{25}-d\right)+s\left(\frac{11t}{15}-t+c+d\right)\\
    +s\left(\frac{2t}{15}-c\right)+s\left(\frac{13t}{30}-d\right)+s\left(1-\frac{17t}{30}-t+c+d\right).
  \end{multline*}
  Rearrange $t<\frac{30}{47}$ and we get
  $$\frac{30-22t}{75}>\frac{t}{3}\geq c.$$
  Rearrange $t<\frac{30}{47}$ and we also get
  $$\frac{15-11t}{25}>\frac{t}{2}\geq d.$$
  Similarly, rearrange $t<\frac{30}{47}$ and we get
  $$1-\frac{17t}{30}>t\geq t-c-d.$$
  Hence,
  $$6W_A\left(P_A,p'_B\right)=3+s\left(-\frac{4t}{15}+c+d\right)+s\left(\frac{2t}{15}-c\right)+s\left(\frac{13t}{30}-d\right)$$
  \begin{itemize}
    \item If $c<\frac{2t}{15}$,
      \begin{itemize}
        \item and if $d\geq\frac{13t}{15}$, then $c+d>\frac{4t}{15}$. So
          \begin{align*}
            6W_A\left(P_A,p'_B\right)&=3+s\left(-\frac{4t}{15}+c+d\right)+s\left(\frac{2t}{15}-c\right)+s\left(\frac{13t}{30}-d\right)\\
            &\geq3+1+1+0\\
            &=5.
          \end{align*}
        \item and if $d<\frac{13t}{15}$, then
          \begin{align*}
            6W_A\left(P_A,p'_B\right)&=3+s\left(-\frac{4t}{15}+c+d\right)+s\left(\frac{2t}{15}-c\right)+s\left(\frac{13t}{30}-d\right)\\
            &\geq 3+0+1+1\\
            &=5.
          \end{align*}
      \end{itemize}
    \item If $c=\frac{2t}{15}$, then, $d\leq \frac{t-c}{2}=\frac{13t}{30}$, and $c+d\geq2c=\frac{4t}{15}$.
      \begin{itemize}
        \item If $d=\frac{13t}{30}$, then $t-c-d=\frac{13t}{30}$. So
          \begin{align*}
            6W_A\left(P_A,p'_B\right)&=3+s\left(-\frac{4t}{15}+c+d\right)+s\left(\frac{2t}{15}-c\right)+s\left(\frac{13t}{30}-d\right)\\
            &=3+1+ \frac{1}{2} + \frac{1}{2} \\
            &=5.
          \end{align*}
        \item If $d<\frac{13t}{30}$, then
          \begin{align*}
            6W_A\left(P_A,p'_B\right)&=3+s\left(-\frac{4t}{15}+c+d\right)+s\left(\frac{2t}{15}-c\right)+s\left(\frac{13t}{30}-d\right)\\
            &\geq3+\frac{1}{2}+ \frac{1}{2} + 1\\
            &=5.
          \end{align*}
      \end{itemize}
    \item If $c>\frac{2t}{15}$, then, $c+d\geq2c>\frac{4t}{15}$, so
      \begin{align*}
        6W_A\left(P_A,p'_B\right)&=3+s\left(-\frac{4t}{15}+c+d\right)+s\left(\frac{2t}{15}-c\right)+s\left(\frac{13t}{30}-d\right)\\
        &\geq3+1+1+0\\
        &=5.
      \end{align*}
  \end{itemize}
  Hence, $W_A(P_A,p'_B)\geq\frac{5}{6}$ for any pure strategy $p'_B$.

  To conclude, player $A$ cannot increase payoff above $\frac{5}{6}$ by changing strategy, 
  and player $B$ can also not increase payoff above $\frac{1}{6}$ by changing strategy. 
  So $P_A,P_B$ form a Nash equilibrium and the equilibrium payoff, $W_3(t)$, is $\frac{5}{6}$.
\end{proof}

\begin{remark}
  To guarantee a Nash equilibrium, player $A$ can play any strategy 
  \[
  P_A=\left\{\left(\left(a,b,c\right),\frac{1}{2}\right),\left(\left(d,e,f\right),\frac{1}{2}\right)\right\}
  \]
  satisfying
  $$a+b+c=1,a\leq b\leq c, d\leq e\leq f, d+e+f=1,$$ 
  $$ a>\frac{t}{3}, b>\frac{t}{2}, f> t,$$
  $$\text{and } 2d+c\geq t, d+2e\geq t.$$
  The strategy $\left\{\left(\left(\frac{30-22t}{75},\frac{15-11t}{25},\frac{11t}{15}\right),\frac{1}{2}\right),\left(\left(\frac{2t}{15},\frac{13t}{30},1-\frac{17t}{30}\right),\frac{1}{2}\right)\right\}$ is only one of the possible ones.
\end{remark}

\begin{theorem}
  \label{thm:w3t2/3}
  $W_3(\frac{2}{3})\le \frac{4}{5}$.	
\end{theorem}

\begin{proof}

  Let $B$ play the strategy
  \begin{multline*}
    P_B=\left\{
    \left(\left(0,\frac{1}{16},\frac{29}{48}\right),\frac{1}{5}\right),
    \left(\left(0,0,\frac{2}{3}\right),\frac{1}{5}\right), 
    \left(\left(\frac{1}{16},\frac{1}{16},\frac{13}{24}\right),\frac{1}{5}\right),\right. \\
    \left.
    \left(\left(\frac{1}{8},\frac{13}{48},\frac{13}{48}\right),\frac{1}{5}\right),
    \left(\left(\frac{5}{24},\frac{11}{48},\frac{11}{48}\right),\frac{1}{5}\right)
    \right\}
  \end{multline*}
  and let $A$ play any pure strategy $p$, then we verified using a computer that the payoff for $A$, $W_A\left( p, P_B \right)$, is at most $\frac{4}{5}$.
\end{proof}

\begin{theorem}
  \label{thm:w3t5/6}
  $W_3(\frac{5}{6})\ge\frac{2}{3}$.
\end{theorem}
\begin{proof}
  Let $A$ play the strategy 
  $P_A=\left(\frac{1}{6},\frac{1}{3},\frac{1}{2}\right)$,
  and let $B$ play any pure strategy $p$, then we verified using a computer that the payoff for $A$, $W_A\left( P_A, p \right)$, is at least $\frac{2}{3}$.
\end{proof}

Note that in this case where $t\neq1$ and in the general case for $n=2$ (\Cref{sec:w2t}), there are Nash equilibrium strategies with atoms. This is behavior very different from what we saw in \Cref{sec:uniqueCDFACB113}, and also very different from what we proved about Nash equilibria of the game $ACB(1,1,n)$ where $n\ge3$ in \Cref{thm:noAtom}.
\section{Open problems} \label{sec:open}

Still much remains unknown about the Asymmetric Colonel Blotto game in the general case. For example, what would a Nash equilibrium for the game $ACB(1,1,4)$ look like? Or a Nash equilibrium for the game $ACB(1,1,n)$ where $n\ge 5$?
The methods used to prove the uniqueness of the marginal distributions of the game $ACB(1,1,3)$ in \Cref{sec:uniqueCDFACB113} cannot be used here since these methods only prove the uniqueness of the marginal distributions given a Nash equilibrium strategy, so these methods do not work when we cannot find a Nash equilibrium in the first place.
It is hard to find a Nash equilibrium for the game $ACB(1,1,4)$, although one can be approximated by means of computer simulation. What's more, we can show that the marginal distributions of Nash equilibrium strategies cannot be uniform. This makes it difficult to guess the correct Nash equilibrium strategy.

Another problem is to determine how the unique equilibrium payoff varies in the game $ACB(1,t,n)$ as $t$ varies continuously in the general case. As we have shown in \Cref{sec:w2t}, $W_2(t)$ is locally constant and discontinuous as a function of $t$. This is quite a surprising result, as it indicates that there are phase changes in the game $ACB(1,t,2)$ as $t$ changes.
Our partial results in \Cref{sec:w3t} also indicate that $W_3(t)$ is a discontinuous function (\Cref{thm:w3t6/11} and \Cref{thm:w3t6/11_18/31}). Computer simulation of discrete cases also indicates that sometimes it is not differentiable where the function itself is continuous. Maybe the phase changes in this case correspond to discontinuous jumps in the equilibrium strategies. This can be illustrated by the drastic difference between the equilibrium strategies in \Cref{sec:w3t} and those in \Cref{sec:uniqueCDFACB113}.
Is it possible to find all the critical values of $t$ where these phase changes occur?

Yet another fundamental question left unanswered is the existence of Nash equilibria for the game $ACB(X_A,X_B,n)$ in the general case. We have discussed Nash equilibria in special cases, but we have not given a proof that guarantees the existence of Nash equilibria in the general case.

\bibliographystyle{alpha}
\bibliography{cb}

\begin{thebibliography}{BKdV96}

\bibitem[BKdV96]{BayeEtAl}
Michael~R. Baye, Dan Kovenock, and Casper~G. de~Vries.
\newblock The all-pay auction with complete information.
\newblock {\em Econom. Theory}, 8(2):291--305, 1996.

\bibitem[Bor53]{Borel53}
Emile Borel.
\newblock The theory of play and integral equations with skew symmetric
  kernels.
\newblock {\em Econometrica}, 21:97--100, 1953.

\bibitem[Rob06]{Roberson2006}
Brian Roberson.
\newblock The {C}olonel {B}lotto game.
\newblock {\em Econom. Theory}, 29(1):1--24, 2006.

\bibitem[Skl59]{Sklar59}
M.~Sklar.
\newblock Fonctions de r\'epartition \`a {$n$} dimensions et leurs marges.
\newblock {\em Publ. Inst. Statist. Univ. Paris}, 8:229--231, 1959.

\bibitem[SS83]{SchweizerSklar}
B.~Schweizer and A.~Sklar.
\newblock {\em Probabilistic metric spaces}.
\newblock North-Holland Series in Probability and Applied Mathematics.
  North-Holland Publishing Co., New York, 1983.

\end{thebibliography}

\end{document}